\documentclass[a4paper,11pt]{article}

\usepackage{graphicx}
\usepackage{amsfonts}
\usepackage{amssymb}
\usepackage{slashed}
\usepackage{braket}
\usepackage{mathrsfs}
\usepackage[margin=1in]{geometry}
\usepackage{siunitx}
\usepackage{hyperref}
\usepackage{wrapfig}
\usepackage{mathtools}
\usepackage[standard]{ntheorem}
\usepackage{lipsum}
\usepackage{csquotes}
\usepackage{float}
\usepackage{bm}

\newtheorem{conjecture}[theorem]{Conjecture}

\usepackage[backend=biber,style=alphabetic,language=auto,maxbibnames=99]{biblatex}
    
\AtBeginDocument{

\DeclareSymbolFont{bbold}{U}{bbold}{m}{n}
\DeclareSymbolFontAlphabet{\mathbbold}{bbold}
\DeclareMathOperator{\Tr}{Tr}

\DeclareMathOperator{\Perm}{Perm}
\DeclareMathOperator{\poly}{poly}

\newcommand{\basis}[1]{\vec{\mathbf{e}_{{#1}}}}
\DeclareSymbolFont{extraitalic}      {U}{zavm}{m}{it}
\DeclareMathSymbol{\stigma}{\mathord}{extraitalic}{168}
\DeclareMathSymbol{\Stigma}{\mathord}{extraitalic}{167}
} 

\hfuzz=\maxdimen \tolerance=10000 \hbadness=10000

\title{Inapproximability of Positive Semidefinite Permanents and Quantum State Tomography}

\ifthenelse{1=1}{
	\author{Alex Meiburg \\ \footnotesize{ameiburg@ucsb.edu, University of California, Santa Barbara}}
	\date{}
}{
	\author{\textit{\small{[Redacted Author Name]}}}
	\date{}
}

\begin{document}
\maketitle

\begin{abstract}
Matrix permanents are hard to compute or even estimate in general. It had been previously suggested that the permanents of Positive Semidefinite (PSD) matrices may have efficient approximations. By relating PSD permanents to a task in quantum state tomography, we show that PSD permanents are NP-hard to approximate within a constant factor, and so admit no FPTAS (unless P=NP). We also establish that several natural tasks in quantum state tomography, even approximately, are NP-hard in the dimension of the Hilbert space. These state tomography tasks therefore remain hard even with only logarithmically few qubits.
\end{abstract}

\section{Introduction}
\subsection{Background}
The permanent is a classical problem of intense interest in the study of counting problems. For a matrix $M \in \mathbb{C}^{n\times n}$, the permanent is defined as
\begin{equation}\label{eqn:permDef}
\Perm(A)=\sum_{\sigma\in S_n}\prod_{i=1}^n a_{i,\sigma(i)}
\end{equation}
summing over all products of permutations of rows and columns. While directly evaluating the expression in Eq \ref{eqn:permDef} takes $O(n!)$ time, Ryser's formula\cite{Ryser63} gives an $O(2^n n)$ time algorithm. Valiant showed in 1989 that computing the permanent exactly is \#P-hard, even for 0-1 matrices\cite{VALIANT1979189,bendor}. However, it is amenable to efficient approximation in particular settings. In 2001, Jerrum, Sinclair and Vigoda\cite{jerrum04} gave a fully-polynomial randomized approximation scheme (FPRAS) for permanents of nonnegative matrices. In 2002, Gurvits and Samorodnitsky\cite{Gurvits2002} gave a polynomial time $e^n$ multiplicative approximation to PSD mixed discriminants, which included permanents of nonnegative matrices as a special case.
\par When the matrix is Hermitian positive semidefinite (HPSD, or if purely real, PSD), the permanent is necessarily nonnegative, and this offers hope of efficient multiplicative approximation. HPSD permanents are of particular interest to the quantum information community - for reasons unrelated to quantum state tomography, but rather related to thermal BosonSampling experiments\cite{Tamma2014,Keshari2015,Kim2020}. It is known that by Stockmeyer counting\cite{Grier18,Keshari2015,Stockmeyer83} computing multiplicative approximations to PSD permanents is contained in $\textsf{FBPP}^\textsf{NP}$. In 1963, Marcus\cite{Marcus63} observed that the product of the diagonal of a PSD matrix immediately gives an $n!$ approximation ratio to the permanent. In 2017, \cite{Anari2017} gave a polytime approximation to PSD permanents within a ratio of $c^n$ with $c = e^{1+\gamma} \approx 4.85$. \cite{Yuan2021} described a similar approach with the same approximation ratio. \cite{Chakhmakhchyan2017} and \cite{barvinok2020remark} gave algorithms for approxmation when the spectrum of the matrix is small in radius, that is, when $\lambda_{min}/\lambda_{max}$ is not too small.

\subsection{Main Results}
Our main result is to show that there is no efficient approximation of PSD permanents. This can be stated as the absence of a fully-polynomial time approximation scheme (FPTAS) or fully-polynomial randomized approximation scheme (FPRAS).
\begin{corollary}[{\textnormal{\em of Thm \ref{thm:npApproxPerm}}}]
There is no FPTAS for HPSD permanents unless \textnormal{\textsf{P}=\textsf{NP}}, and there is no FPRAS for HPSD permanents unless \textnormal{\textsf{RP}=\textsf{NP}}.
\end{corollary}
More precisely, we show that it is \textsf{NP}-hard to approximate within a particular subexponential factor.
\begin{theorem}[{\textnormal{\em Thm \ref{thm:npApproxPerm}, restated}}]
For any constant $\epsilon > 0$, it is \textsf{NP}-hard to approximate the permanent of $n\times n$ HPSD matrices within a factor of $2^{n^{1-\epsilon}}$.
\end{theorem}
In Section \ref{sec:realMats}, we show that these theorems also hold for (purely real) PSD matrices.
\par Our work provides a lower bound on the difficulty of approximating PSD permanents, that almost matches known upper bounds. The algorithm of \cite{Anari2017} shows that the singly exponential approximation ratio $4.85^n$ is possible within polynomial time, while we show that subexponential approximation ratio $2^{n^{1-\epsilon}}$ is intractable. This primarily leaves the question whether $(1+\epsilon)^n$ is polynomial-time computable for any $\epsilon > 0$. The algorithms of \cite{Chakhmakhchyan2017} and \cite{barvinok2020remark} fail on the hard instances that we construct: the matrices we construct are highly rank deficient, and therefore have $\lambda_{min}=0$.
\par We arrived at our hard instances via a problem in quantum state tomography. If a matrix $M$ is positive semidefinite, then it has a matrix square root $VV^\dagger = M$, and we show that
$$\Perm(M)=\frac{(d+n-1)!}{2\pi^n} \int_{\mathbf{v} \in \mathbb{C}^n,\,|v|=1} \prod_{i} |\mathbf{v}\cdot V_i|^2.$$
This last expression occurs naturally in the context of tomography, where the rows $V_i$ of $V$ correspond to an observation history. We analyze the problem by first establishing a concentrating construction (Lemmas \ref{Ulemma} and \ref{inUlemma}). When $V_i$ contains many copies of basis vectors and vectors of the form $\frac{\basis j\pm i\basis k}{\sqrt 2}$, the integral concentrates at the points (up to a phase) of an appropriately scaled hypercube:
$$\int_{\mathbf{v} \in \mathbb{C}^n,\,|v|=1} \prod_{i} |\mathbf{v}\cdot V_i|^2  \approx C \sum_{v \in \{-1,+1\}^d} \prod_{i} |\mathbf{v}\cdot V_i|^2 $$
for some constant $C$ that depends only on $d$ and $n$. This concentration will let us relate permanents to combinatorial problems (Lemma \ref{Glemma}), specifically counting solutions to \textsf{Not-All-Equal-3SAT}, and ultimately let us prove hardness.
\par The connection to quantum state tomography means we also get results about the hardness of estimating quantum states given measurements. For a quantum system with Hilbert space dimension $n$ and $\poly(n)$ observations, the {\em maximum pure state likelihood} is the highest likelihood of those observations attainable over any pure state $\ket{\psi}$.

\begin{theorem}[{\textnormal{\em Thm \ref{thm:MLEhard}, informal}}]
For any constant $\epsilon > 0$, the following task is \textsf{NP}-complete: given a series of quantum observations, find a pure state with likelihood a factor of $2^{n^{1-\epsilon}}$ of the maximum pure state likelihood.
\end{theorem}

In other words, there is no FPRAS for maximum likelihood estimation (MLE) quantum state tomography unless \textsf{RP}=\textsf{NP}. We have similar statements about the \textsf{NP}-hardness of computing the Bayesian average state and Bayesian average observables (Theorem \ref{thm:stateHard}). These results are unusual in that they imply exponential difficulty in dimension $n$ in the Hilbert space $\mathbb{C}^n$. Most quantum problems are only considered tractable if they have efficient algorithms in the number of particles $q = \log(n)$, and have trivially polynomial solutions in $n$; whereas we show that (assuming ETH\cite{Impagliazzo1999}) quantum state tomography takes time exponential in $n$.
\par We stress that although our work has connections to quantum information through BosonSampling and tomography, our discussion of complexity is focused on classical computers. The \textsf{NP}-hardness are statements about classical hardness, and the algorithm described in section \ref{sec:fixDAlgo} for tomography in fixed dimension is a polynomial time {\em classical} algorithm. Unless \textsf{NP}$\subseteq$\textsf{BQP} however, our results rule efficient permanent computations on quantum computers as well.

\section{Key ideas of the proof}
We start with a lemma relating symmetric, multilinear functions to permanents. Similar lemmas have appeared in \cite{barvinok16book,barvinok2020remark}, and they can broadly be viewed as alternate forms of Wick's Theorem \cite{ZVONKIN1997281}. 
\begin{lemma}\label{symLemma}\normalfont 
Suppose $f : (\mathbb{C}^d)^{2n} \to \mathbb{R}$ is a function of $2n$ vectors of dimension $d$, with the properties:
\begin{itemize}
\item Multilinear in its first $n$ arguments:
$$f(v_1,\dots,\alpha v_i + \beta v_i',\dots ) = \alpha f(v_1,\dots, v_i,\dots) + \beta f(v_1,\dots,v_i',\dots)$$
\item Conjugate multilinear in its latter $n$ arguments:
$$f(v_1,\dots,\alpha v_{n+i} + \beta v_{n+i}',\dots) = \alpha^* f(v_1,\dots, v_{n+i},\dots) + \beta^* f(v_1,\dots, v_{n+i}',\dots)$$
\item Symmetric in its first $n$ arguments, and its latter $n$ arguments:
$$f(v_1,v_2,\dots; v_n,v_{n+1},\dots) = f(v_{\sigma(1)},v_{\sigma(2)},\dots, v_{\tau(n)}, v_{\tau(n+1)})$$
\item Invariant under unitary change of basis: for any unitary $U \in \mathbb{C}^{d\times d}$,
$$f(v_1,v_2,\dots; v_{n}, v_{n+1},\dots) = f(Uv_1,Uv_2,\dots; U v_n, U v_{n+1} )$$
\end{itemize}
Then $f$ is determined up to an overall constant $C$ by the formula,
\begin{equation}\label{fSym}
f(v_1,\dots; v_{n},\dots) = C \Perm(A_{ij}),\quad\textrm{ where } A_{ij} = v_i \cdot v_j^*
\end{equation}
and the constant $C$ can be determined by
\begin{equation}\label{cFormula}
C = \frac{f(\basis1,\basis1,\basis1,\dots)}{n!}
\end{equation}
where $\basis1$ is the unit basis vector in the first coordinate.
\end{lemma}
\begin{proof}
Because $f$ is invariant under a unitary change of basis, $f$ can only depend on its inputs through inner products of vectors, $\langle v_i,v_j\rangle$. Since $f$ is multilinear, it can be written as a sum of terms $t_k$, where each $t_k$ is a product of terms from the vectors. The separate linearity and conjugate linearity means that we can only have inner products of covariant (first $n$) and contravariant (latter $n$) vectors. This means every term in the sum must be some product of the form $\prod_{i \in [n]} v_i \cdot v_{n+\sigma(i)}^*$ for some permutation of $n$. Then by symmetry of the arguments, all pairs must occur in the same relation to either, so all pairings must occur equally. This leaves only a single form, the result above.
\par Computing $C$ can be found by substituting in $\basis1$ in \ref{fSym} so that all dot products become 1. The permanent of the all-1's matrix is just $n!$, so this becomes the normalizing factor.
$\square$
\end{proof}
This lets us relate the permanent to a particular integral over unit-norm complex vectors:
\begin{theorem}\label{thm:permIsSphere}
For any $L, R\in\mathbb{C}^{d\times n}$ be complex matrices, denoting the $k$th row as $L_k$,
$$\int_{\mathbf{x} \in \mathbb{C}^n,\,|x|=1} \left(\prod_{k} \mathbf{x}^\dagger L_k\right)\left(\prod_{k} R_k^\dagger \mathbf{x} \right) = \frac{2\pi^n}{(d+n-1)!} \Perm(LR^\dagger)$$
Note that when $L=R$, the product in the integral becomes $\prod_k |\langle L_k, \mathbf{v}\rangle|^2$, and the product $M = L L^\dagger$ is PSD.
\end{theorem}
\begin{proof}
Viewing the left side as a function $f$ of the $n$ rows of $L$ and $R$, we can see that it satisfies all the hypotheses of Lemma \ref{symLemma}. It is linear in each row of $L$, conjugate linear in each row of $R$, and symmetric under permuting the rows of $L$ or the rows of $R$. It is also invariant under a unitary change of basis:
$$f(UL,UR) = \int_{\mathbf{x} \in \mathbb{C}^n,\,|x|=1} \prod_{k} \mathbf{x}^\dagger (UL_k)\,\,\prod_{k} (UR_k)^\dagger \mathbf{x} = \int_{\mathbf{x} \in \mathbb{C}^n,\,|x|=1} \prod_{k} (U^\dagger \mathbf{x})^\dagger L_k\,\,\prod_{k} R_k^\dagger (U^\dagger \mathbf{x})$$
$$ = \int_{\mathbf{u} \in \mathbb{C}^n,\,|u|=1} \prod_{k} \mathbf{u}^\dagger L_k\prod_{k} R_k^\dagger \mathbf{u} = f(L,R)$$
so that we've used the symmetry of the unit sphere in $\mathbb{C}^n$ to remove the unitary via $\mathbf{u} = U\mathbf{x}$. Setting each $L_k = R_k = \basis1$, the spherical integral can be computed with standard formulae (e.g. \cite{Folland2001}) to find the normalizing constant
$$C = \frac{1}{n!}\int_{\mathbf{x} \in \mathbb{C}^n,\,|x|=1} (\mathbf{x}^\dagger \basis1)^n(\basis1^\dagger \mathbf{x})^n = \frac{1}{n!}\cdot\frac{2pi^dn!}{(d+n-1)!} = \frac{2\pi^n}{(d+n-1)!}.$$
$\square$
\end{proof}
This formula is similar to another well-known expression for the permanent involving Gaussian integrals, and can be understood as a version of Wick's theorem. \cite{barvinok16book,ZVONKIN1997281}

\subsection{Outline of the proof}
Before diving into the proof of hardness itself, we aim to provide some intuition of the construction. We focus on the integral $F = \int_{\mathbf{x}} \prod_k |\langle V_k, \mathbf{x}\rangle|^2$ over a the sphere of unit (complex) vectors, and build up a set of vectors $V$ with desirable properties. The proof will involve gradually adding vectors to a list $V_k$, in turn modifying the integrand $I_V(\mathbf{x}) = \prod_k |\langle V_k, \mathbf{x}\rangle|^2$. This integrand $I_V(\mathbf{x})$ is nonnegative, so there cannot be any cancellation in the integral. Our goal will be only showing that certain regions have exponentially small magnitude, so that only particular regions with appreciable contribution remain, and they are primarily responsible for the overall value of $F$. Then, the magnitude of $F$ will be used to understand the value of $I$ on those particular regions, where large values of $F$ indicate solutions to an \textsf{NP}-hard problem. And since $F$ can be computed by a HPSD permanent, computing that permanent must be hard as well.

\par How are we to choose the $V$ in order to make an interesting function $I_V$? Each vector $V_k$ introduces zeroes on the sphere at all vectors orthogonal to $V_k$. All points approximately orthogonal to $V_k$ will have a very small magnitude, and so contribute very little to the integral. We will start our collection of vectors includes many copies of each standard basis vector $\basis k$. This creates high-degree zeros along each of $d$ distinct perpendicular directions, slicing the sphere so that the only regions with appreciable magnitude form the corners of a cube.

\begin{figure}[H]
\begin{center}
\includegraphics[width=2in]{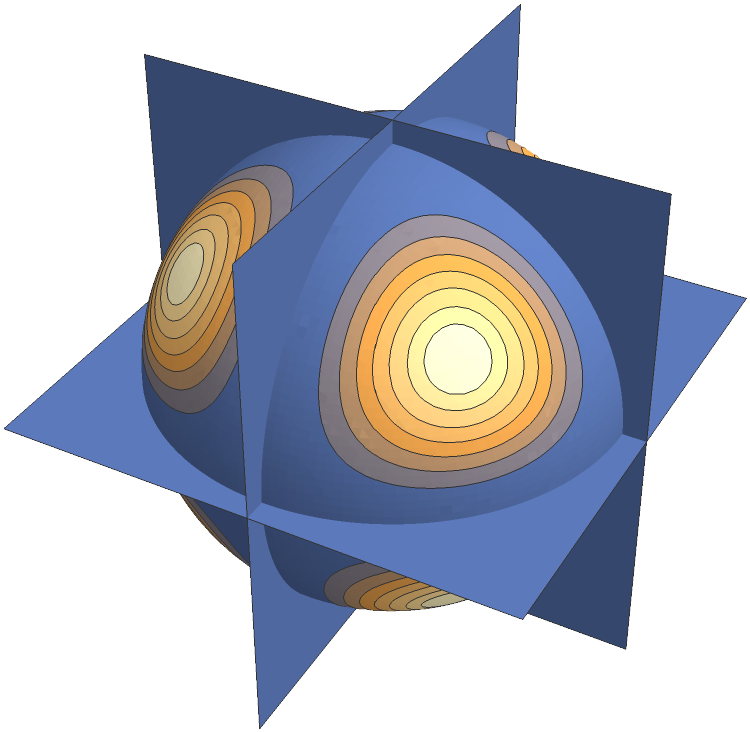}
\end{center}
\caption{Schematic of how we can create ``corners" on the sphere by repeatedly cutting with planes. Blue represents lower magnitude. This shows only purely real $\mathbf{x}$.
\label{sphereCuts}}
\end{figure}

After adding one copy of each basis vector $\basis k$, the magnitude at a given point $\mathbf{x}=(\alpha_1,\dots \alpha_d)$ is the product of the absolute values of its entries in that basis: $I_V(\mathbf{x}) = \prod_j |\alpha_j|$. This is maximized when $|\alpha_j| = |\alpha_k|=\frac{1}{\sqrt{d}}$ for all $j$, $k$. If we then subsequently add several vectors of the form $\frac{\basis j+i\basis k}{\sqrt{2}}$ and $\frac{\basis j-i\basis k}{\sqrt{2}}$, together these rule out a purely imaginary phase between the $j$ and $k$ components, so that the maxima are at $\frac{\basis j\pm \basis k}{\sqrt{2}}$. After adding these two for each $j\neq k$, $I(\mathbf{x})$ will peak near $\mathbf{x} = \frac{e^{i\theta}}{\sqrt{d}}(1,\pm 1,\pm 1\dots)$. Up to an overall phase of $\mathbf{x}$, we've focused $I$ to a set of $2^{d-1}$ distinct points. These $2^{d-1}$ circles of ``binarized" vectors form a set $B_0$. To get this, we had to put $d + 2{d \choose 2} = d^2$ vectors into $V_k$. By analogy with quantum information, we will refer to these as the $Z$ vectors and $Y$ vectors respectively. Together, this set of $d^2$ vectors will form one ``basic set" -- ``basic" in the set of ``enforcing the basis".

Once we have our basic vectors to concentrate $I$ at these binarized points $B_0$, we want to add vectors that will penalize some of these $2^{d-1}$ points, so that finding the optimum becomes a search problem over exponentially many points. Our functional $I$ is only sensitive to the relative phase between components of a vector, and not to the signs of the components themselves. This leads us most naturally to the problem of Not-All-Equal 3-Satisfiability, or \textsf{NAE3SAT}.\cite{Schaefer78} So now consider the impact of adding a triple of ``clause vectors",
$$\mathbf{v}_1 = \frac{\basis 1+\basis 2-2\basis 3}{\sqrt{6}}$$
$$\mathbf{v}_2 = \frac{\basis 1-2\basis 2+\basis 3}{\sqrt{6}}$$
$$\mathbf{v}_3 = \frac{-2\basis 1+\basis 2+\basis 3}{\sqrt{6}}.$$
Each is orthogonal to $\frac{1}{\sqrt{3}}(\basis 1+\basis 2+\basis 2)$, in which all the relative signs are positive (or equivalently, all negative). We call this collection of three vectors a ``clause set''. This effectively rules out the possibility of all signs being the same. There are three not-all-equal points (up to phase):
$$\mathbf{p}_1 = \frac{1}{\sqrt{3}}(\basis 1+\basis 2-\basis 3)$$
$$\mathbf{p}_2 = \frac{1}{\sqrt{3}}(\basis 1-\basis 2+\basis 3)$$
$$\mathbf{p}_3 = \frac{1}{\sqrt{3}}(-\basis 1+\basis 2+\basis 3)$$
The three $\mathbf{p}_i$ all have the same squared inner products with the set of $\mathbf{v}_i$, those being $\{\frac 89,\,\frac 2{9},\,\frac 2{9}\}$ in some order, and so all $\mathbf{p}_i$ have an equal $I_V(\mathbf{p}_i) = \frac{32}{729}$.
\par The total effect may be visualized in the following plot:
\begin{figure}[H]
\begin{center}
\includegraphics[width=7in]{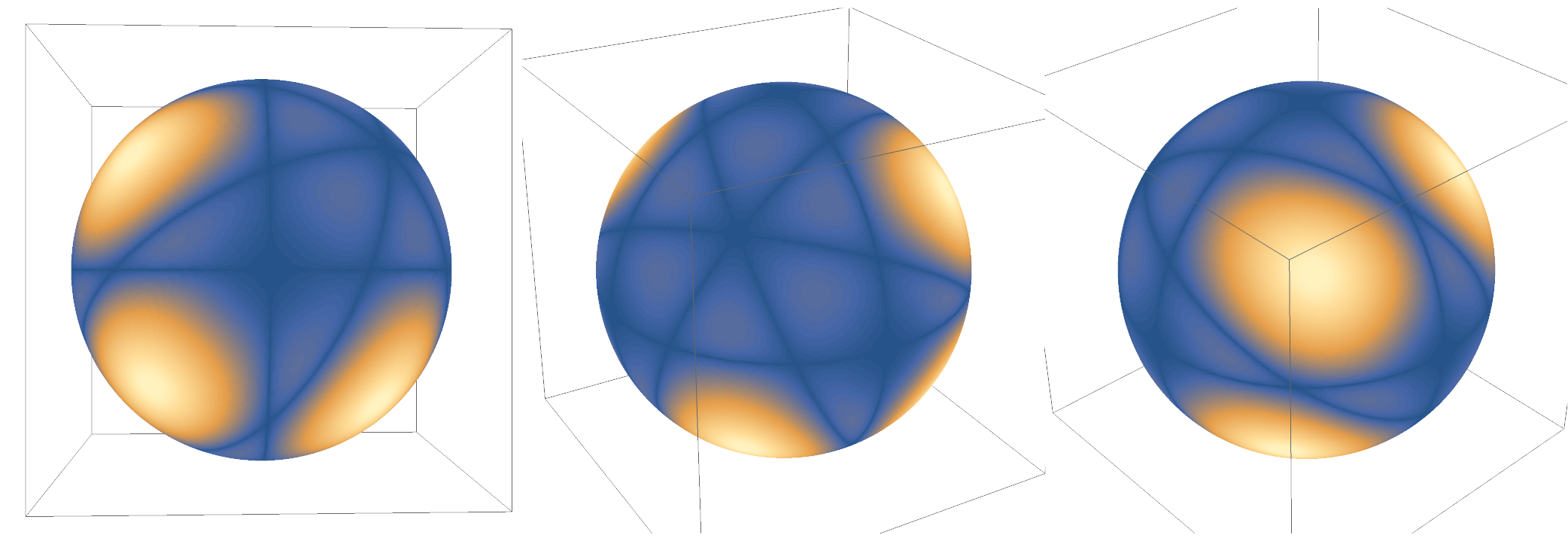}
\end{center}
\label{q3oneclause}
\caption{Three plots of $I_V(\mathbf{x})$. Only real points are plotted, smaller values are blue. As the integrand only depends on points up to an overall phase, all points appear effectively doubled, as $I_V(\mathbf{x}) = I_V(-\mathbf{x})$. There are twelve vectors in $V$. Nine come from a basic set: $\basis 1$, $\basis 2$, $\basis 3$, and $(\basis j\pm i\basis k)/\sqrt{2}$ in six permutations (given by $j,k\in\{1,2,3\}$, $j<k$). The right-angled crosses are due to the first three vectors, dividing the space into eight corners (the $Z$ vectors). The last three vectors in $V$ are a clause set: $(\basis j+\basis k-2\basis \ell)/\sqrt{6}$ (in all 3 permutations), creating the 6-way intersection shown in the second diagram, eliminating two opposing corners of the eight.}
\end{figure}

Here we look at $d=3$ with one basic set and one clause set. The plot doesn't display states with complex coefficients, but it can be verified that the maxima have all real phases. Without the clause vectors, there would be eight high amplitude points. The first subplot shows the effect of the clause most directly: out of the four points (up to sign) in $B_0$, one of them -- the top right corner -- has been eliminated. The excluded option has three zero planes running through it.

By adding appropriate clause sets, the only remaining points with large values will be those satisfying an \textsf{NAE3SAT} problem, which is \textsf{NP}-hard. The other points will be too small to contribute to the integral, so that evaluating the integral tells us about the satisfiability of the \textsf{NAE3SAT} problem. With the outline complete, we now begin the steps of the proof, starting with the concentration.

\section{Proof of Hardness}
\subsection{Concetration}
After one basic set, each point in $B_0$ has a value $I_V(\mathbf{x})$ of $1/d^{d^2}$ (by direct calculation). We would like to show that any state far away from $B_0$ has a significantly lower value. For this reason, and with the intuition that the integrand $I_V$ represents likelihood values, we talk about {\em relative} values of $I_V$. By the value of $I_V(a)$ relative to $I_V(b)$, we simply mean $I_V(a)/I_V(b)$.
\par Any unit vector $\mathbf{x} \in \mathbb{C}^d$ can be written as
$$\mathbf{x} = \frac{e^{i\Theta}}{\sqrt{d}} \sum_{k=1}^d \sqrt{\alpha_k} e^{i\pi(\theta_k+n_k)} \vec{\mathbf{e}_k}$$
where $\Theta$, $\alpha_k$ and $\theta_k$ are all real, $\alpha_k\ge 0$, $\sum_k \alpha_k = d$, $\theta_1=0$, and all $\Theta$, $\theta_k \in [-1/2,1/2])$, and $n_k \in \{0,1\}$. The $\vec \alpha$, $\vec \theta$, and $\vec n$ respectively indicate the amplitudes, phases relative to the first component, and signs of the real part. This polar representation is unique except for when one of the $\alpha_0 = 0$, which is a measure-zero set. Accordingly, we can neglect this measure zero set in subsequent discussions of the integral $\int I_V(x)$ -- as we could otherwise arbitrarily set $I_V(x) = 0$ there without modifying the integral.
\begin{lemma}\label{Ulemma}
Let $\mathbf{x}$ be a unit vector with polar representation $\Theta$, $\vec \alpha$, $\vec\theta$, and $\vec n$. Let $\epsilon_\alpha$ be the distance of $\vec\alpha=(\alpha_1,\dots \alpha_d)$ from $\vec 1$. Then when $V$ is one basic set, the value of $I_V(\mathbf{x})$ relative to any point in $B_0$, is at most $1-\frac{\epsilon_\alpha^2}{4d}$. If $\epsilon_\alpha \le 1/2$, then the likelihood is also at most $1-3\theta_i^2$ for all components $\theta_i$ of $\vec \theta$.
\end{lemma}
\begin{proof}
Then $B_0$ consists of the points with $\alpha_k = 1$ and $\theta_k$ is an integer. If $\mathbf{x}$ has significant distance from all elements of $B_0$, then either the amplitudes $\alpha_k$ or phases $\theta_k$ must differ significantly from these conditions. The likelihood after of the measurements is
$$L(\psi) = \left(\prod_k \left|\sqrt{\frac{\alpha_k}{d}}\right|^2 \right)\left(\prod_{j\le k}\left|\frac{
\sqrt{\alpha_j} e^{i\pi(\theta_j+n_j)}+i\sqrt{\alpha_k} e^{i\pi(\theta_k+n_k)}
 }{\sqrt{2d}}\right|^2\, \left|\frac{
\sqrt{\alpha_j} e^{i\pi(\theta_j+n_j)}-i\sqrt{\alpha_k} e^{i\pi(\theta_k+n_k)}
 }{\sqrt{2d}}\right|^2\right)$$
$$= \left(\prod_k \frac{\alpha_k}{d^d} \right)\left(\prod_{j\le k} \frac{\alpha_j^2 + \alpha_k^2 + 2 \alpha_j \alpha_k \cos(2\pi (\theta_j - \theta_k + n_j - n_k))}{4d^2}\right)$$
$$= \frac{1}{d^d(2d)^{d^2-d}} \left(\prod_k \alpha_k \right)\left(\prod_{j\le k} \alpha_j^2 + \alpha_k^2 + 2 \alpha_j \alpha_k \cos(2\pi (\theta_j - \theta_k))\right)$$
The first factor coming from the $Z$ vectors $\basis k$ in the basic set, and the last two factors coming from the $Y$ vectors $\frac{\basis j\pm i\basis k}{\sqrt 2}$, for each $j < k$, in the basic set.

\par The first step is to bound the likelihood in terms of the magnitudes $\alpha_k$. Looking at the effect of the $Z$ vectors, $\prod_k^d \alpha_k$, we have a convex function on the standard $(d-1)$-simplex $\sum \alpha_k = d$. It is clearly maximized at $\vec\alpha_{opt} = (1,1,1,\dots 1)$, where it evaluates to 1. Suppose that our state $\ket{\psi}$ has an associated $\alpha$-vector, $\vec{\alpha} = (\alpha_1,\dots \alpha_d)$ is a distance at least $\epsilon_\alpha$ away from $\vec{\alpha}_{opt}$, and that $\epsilon_\alpha \le 1$. Then one of the coordinates must be at least $\epsilon_\alpha/\sqrt{d}$ away from 1. With generality, let this coordinate be $\alpha_1$. If $\alpha_1 \le 1 - \epsilon_\alpha/\sqrt{d}$, then the greatest the likelihood could still be is when the other $\alpha_k$ are all equal at $1 + \epsilon_\alpha/\sqrt{d}(d-1)$. Multiplying these together, the resulting likelihood is upper-bounded by $1-\frac{\epsilon_\alpha^2}{2(d-1)}$. If $\alpha_1$ has instead been increased so that $\alpha_1 \ge 1 + \epsilon_\alpha/\sqrt{d}$, then the likelihood is maximized when the other $\alpha_k$ are all equal at $1 - \epsilon_\alpha/\sqrt{d}(d-1)$. Multiplying these together, the resulting likelihood is upper-bounded by $1-\frac{\epsilon_\alpha^2}{4d}$. Since the latter of these bounds is looser, we see that any state whose $\vec \alpha$ is at least $\epsilon_\alpha$ away from the all-ones vector has a likelihood at most $1 - \frac{\epsilon^2}{4d}$ in these measurements.
\par This gives bounds on the $Z$ vectors' contribution to the likelihood. To keep this bound when the $Y$ vectors are added, we need to check that they are also maximized at $\vec\alpha = \vec 1$. Each factor 
$$\prod_{j\le k} \alpha_j^2 + \alpha_k^2 + 2 \alpha_j \alpha_k \cos(2\pi (\theta_j - \theta_k))$$
is maximized when $\theta_j - \theta_k$ is an integer, at which point it becomes $\prod_{j\le k} (\alpha_j+\alpha_k)^2 = \left(\prod \alpha_j + \alpha_k\right)^2$. This is in turn globally maximized by $\alpha_j = \alpha_k = 1$, so the error bound on $\vec \alpha$ holds.
\par The next step is to bound the likelihood in terms of the $\vec\theta$. We only care about the degree to which $\theta_i - \theta_j$ is not an integer, let $r_{ij} = \theta_i - \theta_j$ to the nearest integer, so $r_{ij} \in [-1/2,1/2]$. Given that $\cos(2\pi r) \le 1 - 8r^2$ for all $r\in[-1/2,1/2]$, we have a relative likelihood of
$$\frac{I_V(\mathbf x)}{I_V(B_0)} = \frac{\alpha_j^2 + \alpha_k^2 + 2 \alpha_j \alpha_k \cos(2\pi r_{jk})}{\alpha_j^2 + \alpha_k^2 + 2 \alpha_j \alpha_k} \le \frac{\alpha_j^2 + \alpha_k^2 + 2 \alpha_j \alpha_k (1-8r_{jk}^2)}{\alpha_j^2 + \alpha_k^2 + 2 \alpha_j \alpha_k} = 1 - \frac{16\alpha_j\alpha_k}{(\alpha_j+\alpha_k)^2}r_{jk}^2$$
Let's assume that each $\alpha_j$ is in the interval $[1/2,3/2]$ -- which is implied by them being sufficiently close to the all-ones vector, that is, $\epsilon_\alpha \le 1/2$. Then the expression $\frac{16\alpha_j\alpha_k}{(\alpha_j+\alpha_k)^2}$ is at least 3, so
$$\frac{I_V(\mathbf x)}{I_V(B_0)} \le 1 - 3r_{jk}^2$$
which tells us that every phase $\theta_i$ should be close to $0$ for $I_V$ to be large, or else suffer a $1-3r^2$ penalty in the likelihood. $\square$
\end{proof}
Later we will also need {\em lower} bounds on the likelihood, if we {\em are} in $U$.
\begin{lemma}\label{inUlemma}
If a state $\mathbf{x}$ is within distance $\epsilon \le 0.1$ of some point $b$ in $B_0$, and $V$ is one set of basic vectors, then $\mathbf{x}$ has likelihood at least
$$I_V(\mathbf{x}) \ge \frac{1-2\epsilon d^{5/2}}{d^{d^2}}.$$
\end{lemma}
or in terms of the relative value, $I_V(\mathbf{x})/I_V(B_0) \ge 1-2\epsilon d^{5/2}$.
\begin{proof}
We will again use polar representation for $\mathbf{x}$:
$$I_V(\mathbf{x}) = \frac{1}{d^d(2d)^{d^2-d}} \left(\prod_k \alpha_k \right)\left(\prod_{j\le k} \alpha_j^2 + \alpha_k^2 + 2 \alpha_j \alpha_k \cos(2\pi (\theta_j - \theta_k))\right)$$
If our point $\mathbf{x}$ is within distance $\epsilon < 1$ of $B_0$, then each of the $\alpha_i$ must individually be within $\epsilon \sqrt{d}$ of 1, and each $\theta_i$ satisfies
$$\cos(\pi\theta_i) > \sqrt{1-\epsilon^2} \, \implies \, |\theta_i| < \epsilon/2$$
and so
$$\cos(2\pi(\theta_j-\theta_k)) \ge 1 - \frac{(2\pi(\theta_j-\theta_k))^2}{2} \ge 1 - \frac{(2\pi\epsilon)^2}{2}.$$ 
Then the likelihood is bounded by,
$$L(\psi) \ge \frac{1}{d^d(2d)^{d^2-d}} \left(\prod_k 1-\epsilon \sqrt{d} \right)\left(\prod_{j\le k} (1-\epsilon \sqrt{d})^2 + (1-\epsilon \sqrt{d})^2 + 2 (1-\epsilon \sqrt{d})^2 \left(1 - \frac{(2\pi\epsilon)^2}{2}\right)\right)$$
$$ = \frac{1}{d^d(2d)^{d^2-d}}(1-\epsilon \sqrt{d})^{(d^2+d)/2}\left(4 - 4\pi^2\epsilon^2\right)^{(d^2-d)/2}$$
$$ \ge \frac{1}{d^{d^2}}\left(1-\epsilon\sqrt{d}(d^2+d)/2-\pi^2\epsilon^2(d^2-d)/2\right)$$
If $\epsilon < \sqrt{d}/\pi^2$, which is implied by $\epsilon < 0.1$, then the term $\pi^2\epsilon^2(d^2-d)$ is smaller than $\epsilon\sqrt{d}(d^2+d)/2$, so we can combine the two. We can also bound $d^2+d < 2d^2$.
$$L(\psi) \ge \frac{1}{d^{d^2}}\left(1-2 \epsilon\sqrt{d}(d^2+d)/2\right) = \frac{1-2\epsilon d^{5/2}}{d^{d^2}}$$
$\square$
\end{proof}

Together, these two lemmas establish a form of concentration: points close to $B_0$ have large (lower-bounded) values of $I_V$, and points far from $B_0$ have small (upper-bounded) values of $I_V$.

\subsection{Restricting to neighborhoods of $G$}
Now we consider the effect of clause sets. A clause $\mathsf{C}$ is defined by a triple of integers $(\mathsf{C}_{1}, \mathsf{C}_{2}, \mathsf{C}_{3})$. A point $b \in B_0$ with coordinates $(b_1, b_2, \dots b_d)$, each $b_k = \pm e^{i\Theta}$, is ``good" for the clause $\mathsf{C}$ if $\{b_{\mathsf{C}_1}, b_{\mathsf{C}_{2}}, b_{\mathsf{C}_{3}}\}$ are not all equal. A point in $B_0$ is ``good" for a set of clauses if it is good for each of them, and a point is ``bad" if it is not good. Each clause $\mathsf{C}$ has an associated set of three clause vectors
$$\mathbf{v}_1 = \frac{\basis{\mathsf{C}_1} +\basis{\mathsf{C}_2}-2\basis{\mathsf{C}_3}}{\sqrt{6}}$$
$$\mathbf{v}_2 = \frac{\basis{\mathsf{C}_1}-2\basis{\mathsf{C}_2}+\basis{\mathsf{C}_3}}{\sqrt{6}}$$
$$\mathbf{v}_3 = \frac{-2\basis{\mathsf{C}_1}+\basis{\mathsf{C}_2}+\basis{\mathsf{C}_3}}{\sqrt{6}}.$$

\begin{lemma}\label{Glemma}
Take a clause $\mathsf{C} = (\mathsf{C}_{1}, \mathsf{C}_{2}, \mathsf{C}_{3})$ and let $V$ be its three clause vectors. Nowhere does $I_V$ exceed 1. At any point $\mathbf{x}$ within a distance $\epsilon$ of a good point, $I_V(\mathbf{x}) \ge \frac{32}{27d^3}\left(1-12\epsilon\sqrt{d}\right)$. At any point $\mathbf{x}$ within a distance $\epsilon$ of a bad point, $I_V(\mathbf{x}) \le \frac{4096}{27}\epsilon^6$.
\end{lemma}
\begin{proof}
To see that 1 is an upped bound on $I_V$, note that $I_V$ is a product of dot products of unit vectors, each of which is at most 1, so that $I_V \le 1$.

\par For the second claim, we have a point $\mathbf{x}$ close to a good point $\mathbf{g}$. Since we only care about the value of $I_V$ and the distance between $|\mathbf{x}-\mathbf{g}|$, we may adjust the phase of $\mathbf{x}$ and $\mathbf{g}$ jointly so that $\mathbf{g}$ is entirely real, and all of its entries are $\pm 1$. We decompose $\mathbf{x}$ in the form
$$\mathbf{x} = \alpha \basis{\mathsf{C}_1} + \beta \basis{\mathsf{C}_2} + \gamma \basis{\mathsf{C}_3} + \Delta \mathbf{x}_{\perp}$$
Then the impact of the three clause vectors is,
$$I_V(\mathbf{x}) = \frac{1}{6^3} |\alpha+\beta-2\gamma|^2 \cdot |\alpha-2\beta+\gamma|^2 \cdot |-2\alpha+\beta+\gamma|^2 $$
We seek to bound this value in the vicinity of good points. A good $B_0$ point has not all signs equal. Since we can permute the elements of $\mathsf{C}$ without affecting the value of $I_V$, a general good point $\mathbf{g}$ can be written as
$$\mathbf{g} = \frac{1}{\sqrt d}\left( -\basis{\mathsf{C}_1} +\basis{\mathsf{C}_2}+\basis{\mathsf{C}_3} + \sqrt{d-3}\,\mathbf{g}_{\perp}\right)$$
where $\mathbf{g}_{\perp}$ contains the support on all the other basis vectors. It has $I_V(\mathbf{g}) = \frac{32}{27d^3}$, by direct computation. Then for our other point $\mathbf{x}$ within a distance $\epsilon$ of $\mathbf{g}$, each coordinate must also be within $\epsilon$ of the corresponding coordinate in $\mathbf{g}$. So
$$\Re[\alpha + \beta - 2\gamma] \le \frac{1}{\sqrt{d}}\left((-1+\epsilon \sqrt{d}) + (1+\epsilon \sqrt{d}) - 2(1-\epsilon \sqrt{d})\right) = -2(1-2\epsilon \sqrt{d})/\sqrt{d}$$
and similarly
$$\Re[-2\alpha + \beta +\gamma] \ge \frac{1}{\sqrt{d}}\left(-2(-1+\epsilon \sqrt{d}) + (1-\epsilon \sqrt{d}) + (1-\epsilon \sqrt{d})\right) = 4(1-\epsilon \sqrt{d})/\sqrt{d} \ge 4(1-2\epsilon\sqrt{d})/\sqrt{d}.$$
Putting together the six factors,
\begin{align}
I_V(\mathbf{x}) &= \frac{1}{6^3} |\alpha+\beta-2\gamma|^2 \cdot |\alpha-2\beta+\gamma|^2 \cdot |-2\alpha+\beta+\gamma|^2\\
& \ge \frac{1}{6^3} \Re[\alpha+\beta-2\gamma]^2 \Re[\alpha-2\beta+\gamma]^2 \Re[-2\alpha+\beta+\gamma]^2\\
&\ge \frac{1}{6^3}\frac{32}{27d^3}\times \left(1-2\epsilon\sqrt{d}\right)^6\\
&\ge \frac{1}{6^3}\frac{32}{27d^3}\times \left(1-12\epsilon\sqrt{d}\right)
\end{align}
which is the second claim. For the third claim, take a bad point $\mathbf{h}$ in $B_0$, for which we can correct the phase to put it in the form
$$\mathbf{h} = \frac{1}{\sqrt d}\left( +\basis{\mathsf{C}_1} +\basis{\mathsf{C}_2}+\basis{\mathsf{C}_3} + \sqrt{d-3}\,\mathbf{h}_{\perp}\right)$$
Then for a nearby point only $\epsilon$ away, each coordinate is at most $\epsilon$ away. This means
$$\Re[\alpha + \beta - 2\gamma] \le \left(\frac{1}{\sqrt{d}}+\epsilon\right)+\left(\frac{1}{\sqrt{d}}+\epsilon\right)+\left(\frac{-2}{\sqrt{d}}+2\epsilon\right) = 4\epsilon$$
$$\Im[\alpha + \beta - 2\gamma] \le 4\epsilon$$
$$\implies |\alpha+\beta-2\gamma|^2 \le 32\epsilon^2$$
and similarly for the other two permutations, so that
$$I_V(\mathbf{x}) \le \frac{1}{6^3}(32\epsilon^2)^3 = \frac{4096}{27}\epsilon^6.$$
$\square$
\end{proof}

\subsection{$F = \int_x I_V(\mathbf{x})$ Approximates \#NAE3SAT}
With these bounds, we will be able to relate the number of solutions to a NAE3SAT instance to the integral $F = \int_x I_V(\mathbf{x})$.
\begin{theorem}\label{thm:FappNAE}
Given an instance of NAE3SAT with $d$ variables and $k$ clauses, let the set of vectors $V$ be given by $K_1 = 1600d^7 \ln^2(d)$ copies of basic vectors ($Z$ and $Y$ vectors), together $K_2 = d^2 \ln(d)$ copies of the clause vectors for each clause. For sufficiently large $d$, there is a function $p(n,k)$ such that, if there is at least one solution to the NAE3SAT, $F = \int_x I_V(\mathbf{x}) \ge pd^{-22d}$, and if there are no solutions, $F \le pd^{-d^2}$.
\end{theorem}
\begin{proof}
The theorem will hold if we take $p$ as the value of $I_V$ at a good point, or
$$p = d^{-K_1 d^2} \left(\frac{32}{27d^3}\right)^{K_2}.$$
If the original NAE3SAT instance has a satisfying assignment $(1,0,0,1,\dots)$, there is a corresponding good point
$$\mathbf{g} = \frac{1}{\sqrt{d}}\left(+\basis 1-\basis 2-\basis 3+\basis 4\dots\right)$$
with a large value of $I_V(\mathbf{g})$. Each set of basic vectors introduces a factor of $1/d^{d^2}$ in $I$, and each set of clause vectors introduces a factor of $32/27d^3$. Thus
$$I_V(\mathbf{g}) = d^{-K_1 d^2} \left(\frac{32}{27d^3}\right)^{K_2} = p$$
\par Further, we want to show that around this good point $\mathbf{g}$, there is an appreciable volume with large $I_V$, that will contribute substantially to $F$. Around each good point, take the ball of radius
$$\epsilon_g = \frac{1}{3200d^9(1+d)}.$$
Then by Lemma \ref{inUlemma}, each set of basic observations gives a factor in $I$ of at least
$$I_1 \ge \frac{1-2\epsilon_gd^{5/2}}{d^{d^2}}$$
and by Lemma \ref{Glemma}, each set of clause observations gives a factor at least
$$I_2 \ge \frac{32}{27d^3}(1-12\epsilon_g\sqrt{d})$$
so that the final $I_V$ value of each point in the ball is at least
\begin{align}
I_0 & = I_1^{K_1} I_2^{K_2} \ge p(1-2\epsilon_gd^{5/2})^{K_1}(1-12\epsilon_g\sqrt{d})^{K_2}\\
 &\ge p(1-2\epsilon_gK_1d^{5/2})(1-12K_2\epsilon_g\sqrt{d})\\
 &= p\left(1-2\frac{1}{3200d^9(1+d)}(1600d^7\ln^2(d))d^{5/2}\right)\left(1-12(d^2\ln(d))\frac{1}{3200d^9(1+d)}\sqrt{d}\right)\\
 & \ge p\left(1-\frac{\ln^2 d}{\sqrt d}\right)
\end{align}
This means the total contributed to $F$ by the ball around this good point is then at least $p(1-\ln d/\sqrt d)$ times the volume of this ball around $\mathbf{g}$. The ball is not actually a sphere in $\mathbb{R}^{2d}$, as it lies on the manifold of normalized states, which is curved; it's the intersection of a ball centered at $\mathbf{g}$ and the unit sphere. But since $\epsilon_g < 1/2$, this deformation reduces the volume by less than a factor of 1/2, and then we can use the standard volume of the ball. So the volume obeys
$$\textrm{Vol} \ge \frac{1}{2}\cdot\frac{2(d-1)!(4\pi)^{(d-1)}}{(2d-1)!}\epsilon_g^{2d-1}$$
and a single good point contributes a total likelihood to $p_{norm}$ at least
$$\textrm{Vol}\cdot I_0 \ge pc_1c_2^{-d}d^9d^{-21d}$$
for some particular constants $c_1, c_2 > 1$; the $d^{-21d}$ term clearly dominates the scale for large $d$. For sufficiently large $d$ then we can write
$$F \ge \textrm{Vol}\cdot I_0 \ge pd^{-22d}$$
which establishes the first claim. The second claim concerns when there are no good points. Suppose for contradiction that there is some point $\mathbf{x}$ (not necessarily in $B_0$) so that $I_V(\mathbf{x}) > p/d^{d^2}$. Applying Lemma \ref{Ulemma}, we know that it must have $\epsilon_\alpha = |\vec\alpha - \vec 1| < 0.1/d^2$, otherwise it would have at most
\begin{align}
I_V(\mathbf{x}) \le \left(d^{-d^2}(1-0.1^2/4d^5)\right)^{K_1}
&< d^{-K_1d^2}\exp(-K_1/400d^5)\\
&= d^{-K_1d^2}\exp(-4d^2\ln^2 d) \\
&< d^{-K_1d^2}\exp\left(-4d^2\ln^2 d+d^2\ln d\ln(32/27)\right)\\ 
&= d^{-K_1d^2}\exp\left(-d^2\ln^2 d+d^2\ln d\ln(32/27d^3)\right)\\
&= d^{-K_1d^2}\exp\left(-d^2\ln^2 d+\ln\left(\left(\frac{32}{27d^3}\right)^{K_2}\right)\right)\\
&= d^{-K_1d^2}\left(\frac{32}{27d^3}\right)^{K_2}/{d^{d^2\ln d}} \, = p/d^{d^2 \ln d}\\
&< p/d^{d^2}
\end{align}
Since $\epsilon_\alpha \le 1/2$, we can also apply the second part of Lemma \ref{Ulemma} and check that the all phases $|\theta_i| < 0.1/d$, otherwise our point would have $I_V$ at most
\begin{align}
\left(d^{-d^2}(1-3\theta_i^2)\right)^{K_1} &< \left(d^{-d^2}(1-0.03/d^2)\right)^{K_1}\\
&< \left(d^{-d^2}(1-0.1^2/4d^5)\right)^{K_1}\\
&< p/d^{d^2}
\end{align}
Since the amplitudes are all within $\epsilon_a$ of $1/\sqrt{d}$, and the phases are all within $0.1/d$ of $0$, the point's distance to the nearest point $b$ in $B_0$ is at most
$$\textrm{dist}_{B_0} \le \sqrt{d}\left(\epsilon_a + \left(\frac{1}{\sqrt{d}}+\epsilon_a\right)\left((1-\cos(\theta_i))^2+\sin^2(\theta_i)\right)\right)$$
$$\le \sqrt{d}\left(\frac{0.1}{d^2} + \left(\frac{1}{\sqrt{d}}+\frac{0.1}{d^2}\right)\left(2-2\cos(0.1/d)\right)\right)\le \sqrt{d}\left(\frac{0.1}{d^2} + \frac{2}{\sqrt{d}}\left(0.1/d\right)^2\right)$$
$$\le \frac{0.11}{d^{3/2}}$$
If that point $b$ is bad, then by Lemma \ref{Glemma} our point would have $I_V$ at most
$$d^{-K_1d^2}\left(\frac{4096}{27}\left(\frac{0.11}{d^{3/2}}\right)^6\right)^{K_2} = d^{-K_1d^2}\left(\frac{0.00023}{d^{9}}\right)^{K_2} = d^{-K_1d^2}\left(\frac{32}{27d^3}\right)^{K_2} \left(\frac{0.00023}{(32/27)d^6}\right)^{K_2} $$
$$< p \times 0.0002^{K_2} \le p / 5000^{d^2 \ln d} = p / (5000d)^{d^2} < p/d^{d^2}.$$
We've shown that all points have $I_V \le p/d^{d^2}$. The volume of integration is $S_{2n-1} < 1$, so the total integral $F$ is less than $p/d^{d^2}$.
$\square$
\end{proof}

\subsection{NP Hardness}
We can now prove our main result.
\begin{theorem}\label{thm:npApproxPerm}
For any constant $C<1$, it is \textsf{NP}-Hard to approximate the permanent of an $n\times n$ Hermitian positive semidefinite matrix within a factor of $2^{n^C}$.
\end{theorem}
\begin{proof}
We can reduce from NAE3SAT. Given an NAE3SAT instance on $d$ variables, we can use the set of vectors $V$ described in Theorem \ref{thm:FappNAE} and examine the resulting value $F$. As we have $O(d^9)$ vectors in $V$, the quantity $F$ can be represented as a permanent of a matrix of size $O(d^9)$. The NAE3SAT instance is satisfiable if $F \ge pd^{-22d}$ and unsatisfiable if $F \le pd^{-d^2}$, which can be distinguished if approximating within a factor of $d^{d^2-22d} = O(d^{d^2})$, and so $O(2^{d^2})$ will suffice. If we had an oracle that could approximate permanents of size $n$ PSD matrices within a factor of $2^{n^C}$ for some $C<1$, then we could do the replica trick: take the matrix corresponding to $F$, and repeat it $M = d^{(2-9C)/(1-C)}$ many times along the diagonal. The result is a matrix of size $M d^9$, which is then approximated within a factor of $2^{(Md^9)^C}$. The resulting matrix size $Md^9$ is still $poly(d)$ for any fixed $C$. Then we raise this approximate answer to the power $1/M$ to recover an approximation to the original permanent, and it has error
$$\left(2^{(Md^9)^C}\right)^M = 2^{d^{9C}M^{C-1}} = 2^{d^{9C}d^{2-9C}} = 2^{d^2}$$
which is sufficient to distinguish between satisfiable and unsatisfiable instances. As \textsf{NAE3SAT} is \textsf{NP} hard, so is approximating HPSD permanents with this accuracy.
\end{proof}

This result is complementary to one of Anari et al\cite{Anari2017}, where they show that one {\em can} approximate within a factor of $\exp((1+\gamma+o(1))n)$ where $\gamma$ is the Euler-Mascheroni constant, while we showed that permanents cannot be approximated with subexponential error. Our hard instances circumvent the fast approximation schemes of \cite{barvinok2020remark} and \cite{Chakhmakhchyan2017}, which both have requirements on the spectrum of the matrix, and perform more favorably when $\lambda_{max}/\lambda_{min}$ is smaller. Our instances are of low rank (only rank $d$, which is much larger than the matrix size $n$) so that $\lambda_{min} = 0$.
\par Finally, we conjecture that the reduction above is approximation preserving: that each good point contributes an equal amount of likelihood that can easily be estimated beforehand. Showing this would require tighter error bounds.
\begin{conjecture}
With an appropriate choice of polynomial-scaling $K_1$ and $K_2$, the construction used in Theorem \ref{thm:FappNAE} is an approximation-preserving reduction from \#NAE3SAT to HPSD permanents, such that approximating HPSD Permanents within a factor $C$ is as hard as approximating \#NAE3SAT (or \#3SAT) within a factor $C$.
\end{conjecture}
It is known that by Stockmeyer counting\cite{Grier18,Keshari2015,Stockmeyer83} computing multiplicative approximations to PSD permanents is contained in $\textsf{FBPP}^\textsf{NP}$, and if it is indeed as hard as approximating \#3SAT, it seems unlikely to be much easier than this.

\subsection{Real Matrices}\label{sec:realMats}
The arguments above all involve complex vectors, complex matrices, and integrals over the complex unit sphere. The arguments however can easily be adapted to show that PSD permanents remain hard even for purely real matrices. We could have proved the results only for the real case and this would of course imply hardness for the more general complex case, but the proof for the real case was less symmetric, asthetic, or inuitive than the complex case, which is why we delayed to this section.
\begin{theorem}\label{thm:npApproxPermReal}
For any constant $C<1$, it is \textsf{NP}-Hard to approximate the permanent of an $n\times n$ real positive semidefinite matrix within a factor of $2^{n^C}$.
\end{theorem}
\begin{proof}
The construction proceeds very similarly to above, by reducing from NAE3SAT. However, we now use one dimension more in the space: a $d$-variable NAE3SAT problem is mapped to a $(d+1)$-dimensional spherical integral $\int I(\vec x)$. The clauses are mapped, as before, with $K_2$ many sets of clause vectors, connecting the variables $1$ through $d$ in the original problem with dimensions $1$ through $d$ in the spherical integral $I(x)$. The ``basic sets" still include $K_1$ many instances of the unit vectors $\basis k$ in each basis direction $k \in [d+1]$, what we previously referred to as the $Z$ vectors.
\par The $Y$ vectors were, in preivous proofs, of the form $\frac{\basis{j} \pm i\basis{k}}{\sqrt{2}}$, for $j\neq k$. This was the sole source of complex terms in our vectors, and the reasons the resulting matrices were complex. Instead now we use four copies of each of $\frac{\basis{j} \pm \basis{d+1}}{\sqrt{2}}$. These each softly enforce the constraint that the component of $\vec x$ in the $j$ direction and the $d+1$ direction have relative phase $\pm i$ (that is, $\pm \sqrt{-1}$). Since each $j$ has $\pm i$ relative to $d+1$, this implies that each $j\neq k$ have relative phase $\pm 1$.
\par To make this quantitative and precise, we refer to the proof of Lemma \ref{Ulemma}. The bound of $1-\frac{\epsilon_\alpha^2}{4d}$ applies as before, since the $\basis k$ vectors occur just as before. As proved in Lemma \ref{Ulemma}, if $\theta_j$ and $\theta_{d+1}$ differ by a phase (up to $\pm 1$) of $\Delta\!\theta_j = \theta_j - \theta_{d+1}$, then the likelihood $I(x)$ is reduced by a factor of $1 - 3 \Delta\!\theta_j^2$; since we use each vector eight times, this becomes $(1-3\Delta\!\theta_j^2)^8$. Then for two $j\neq k$, $j,k \le d$, the likelihood is at most
$$(1-3\Delta\!\theta_j^2)^4(1-3\Delta\!\theta_k^2)^4 \le \left(1-3\left(\frac{|\Delta\!\theta_j|+|\Delta\!\theta_k|}{2}\right)^2\right)^{8} \le \left(1-3\left(\frac{|\theta_j - \theta_k|}{2}\right)^2\right)^{8} \le 1 - 3(\theta_j - \theta_k)^2$$
which gives us the same bound on the relative phases as before, so that an analogous statement to Lemma \ref{Ulemma} for our new basis set. The proof of Lemma \ref{inUlemma} holds with few modifications: in the proof above, the $Y$ terms 
$$\prod_{j\le k} \alpha_j^2 + \alpha_k^2 + 2 \alpha_j \alpha_k \cos(2\pi (\theta_j - \theta_k))$$
lead to the a penalty
$$\prod_{j\le k} (1-\epsilon \sqrt{d})^2 + (1-\epsilon \sqrt{d})^2 + 2 (1-\epsilon \sqrt{d})^2 \left(1 - \frac{(2\pi\epsilon)^2}{2}\right) = \left((1-\epsilon \sqrt{d})\left(4 - 4\pi^2\epsilon^2\right)\right)^{(d^2-d)/2}$$
$$ \ge 1-(\epsilon\sqrt{d}+\pi^2\epsilon^2)\frac{d^2-d}{2}.$$
Here instead we have four copies of each phase constraints, but only between $j \le d$ and $d+1$. So the penalty from
$$\prod_{j\le d} \Big(\alpha_j^2 + \alpha_{d+1}^2 + 2 \alpha_j \alpha_{d+1} \cos(2\pi (\theta_j - \theta_{d+1}))\Big)^4$$
becomes
$$\prod_{j\le d} \left((1-\epsilon \sqrt{d})^2 + (1-\epsilon \sqrt{d})^2 + 2 (1-\epsilon \sqrt{d})^2 \left(1 - \frac{(2\pi\epsilon)^2}{2}\right)\right)^4 = \left((1-\epsilon \sqrt{d})\left(4 - 4\pi^2\epsilon^2\right)\right)^{4d}$$
$$ \ge 1-(\epsilon\sqrt{d}+\pi^2\epsilon^2)(4d) \ge 1-(\epsilon\sqrt{d}+\pi^2\epsilon^2)\frac{d^2-d}{2}$$
as before, as long as $d\ge 9$. The resulting conclusion of the lemma that the relative value $I_V(\mathbf x)/I_V(B_0) \ge 1 - 2\epsilon d^{5/2}$ thus still holds.
\par Finally, Lemma \ref{Glemma} remains umodified in this setting, as the form of the clause vectors is unchanged. As all the necessary lemmas hold as before, and the proofs of Theorems \ref{thm:FappNAE} and \ref{thm:npApproxPerm} only care about relative values, they will all hold in the real-valued PSD setting. $\square$
\end{proof}

\section{Quantum State Tomography}
The author initially found the above construction while investigating the worst-case hardness of quantum state tomography, and the hardness implies that several problems in the context of tomography are \textsf{NP}-hard as well.
\par Quantum State Tomography (QST) is the procedure of estimating an unknown quantum state from a set of measurements on an identically prepared ensemble. The procedure can encompass both the choosing of measurement bases as well as estimating the resulting state from the measurements; in adaptive settings, the running estimate is also used to inform future measurement choices\cite{Huszar2012,Quek2021}. We focus on the latter task, of building an estimate of the state. We look at four related forms of what ``estimation" can qualify as:
\begin{enumerate}
\item Finding the Maximum Likelihood Estimator (MLE): the pure state $\rho$ most likely to produce the observations.
\item Finding the Bayesian expected state $\rho_{Avg}$: assuming a prior over the possible pure states, finding the mixed state presenting the mixture of appropriately weighted possible states.
\item Computing the expectation value of some future observation(s).
\item Finding the probability that the unknown state is in fact some particular $\rho_0$. (As there are infinitely many different pure states, we are actually asking for the probability {\em density} at $\rho_0$.)
\end{enumerate}
The first three estimations problems have all been extensively studied with various heuristics. MLE can be attempted by linear inversion\cite{Qi2013,DAuria2009}, iterative search\cite{Lvovsky2004,Rehacek2001}, or even neural networks\cite{Torlai2018}. Bayesian estimation can be accomplished by direct numerical integration\cite{BlumeKohout2010} or particle based sampling\cite{Huszar2012}, possibly with neural networks guiding the particles\cite{Quek2021}. Directly estimating future samples has also been attempted with neural networks\cite{Smith2021} or classical shadows\cite{Shadows1,Shadows2,Shadows3}. The author is not aware of any prior work on computing estimation problem 4.
\par We can show that estimation problems 2, 3, and 4 are essentially as hard as approximating PSD permanents, and that task 1 is also \textsf{NP}-hard. The exponential difficulty (assuming ETH\cite{Impagliazzo1999}) is in fact in the dimension $d$ of the underlying Hilbert space. Many questions in quantum information appear to be ``exponentially" hard, in the sense that it is hard to analyze a system $q$ qubits faster than $O(2^k)$. But here $d=2^q$, so that even when the number of qubits is a logarithmically small $q=\log(d)$, the problem of state estimation remains exponentially hard.

\subsection{Outline of Tomography Results}
Of the four forms above, we focus first on estimation problem 4. Although it is likely the question least relevant to experiment, it is the easiest to manipulate algebraically. We call it \textsc{Quantum-Bayesian-Update}, or simply \textsc{QBU}, and define it in section 5.2. In section 5.3, we give an exponential time algorithm for \textsc{QBU}, showing that it is at least possible. In section 5.4 we show that estmation problems 2, 3, and 4 are equivalent. In section 5.5 we explain \textsc{QBU}'s connection to HPSD permanents, and show it is \textsf{NP}-Hard to  approximate within subexponential error. In section 5.6 we show how the construction of difficult PSD permanents can also be modified shows that the MLE problem (estimation problem 1) is also \textsf{NP}-hard to approximate: it is NP-hard to check the existence of a state with likelihood within a subexponential factor.

\subsection{Quantum Bayesian Update}
We define the \textsc{QBU} problem as follows: given a series of observations $\mathcal{O}_i$ each taken from a copy of $\rho$, and a guess $\rho_0$, what is the probability density that $\rho = \rho_0$? The actual probability of equality is zero -- unless we have some other powerful information about the state -- which is why we ask for the probability density in the space of candidate density matrices.

Bayes' theorem lets us compute the probability density of a true state $\rho$ in terms of the likelihood of the observations $P(\mathcal{O}|\rho)$, a prior belief distribution $P(\rho)$, and the total probability of the sequence of observations $P(\mathcal{O})$. It reads,
$$P(\rho_0|\mathcal{O}) = \frac{P(\mathcal{O}|\rho_0)P(\rho_0)}{P(\mathcal{O})}$$
In order for the equation to be meaningful and not identically zero on both sides, we can read $\rho$ as representing a small volume in the space of density matrices. While there are many natural priors on the space of density matrices, we focus on the case where we know the unknown state $\rho$ is pure. This models, for instance, where we are trying to identify the output of a unitary quantum channel. The most natural prior is then the uniform distribution over all pure states, given by the Haar measure. Then all $P(E)$ are equal. The likelihood of a given observation $\mathcal{O}_i$ is simply $\Tr[\mathcal{O}_i\rho]$, so our goal is to compute

$$P(\rho|\mathcal{O}) = \frac{\prod_{i\in [n]} \Tr[\mathcal{O}_i \rho]}{P(\mathcal{O})}$$

In general $\mathcal{O}_i$ could be operators of any rank, and could belong to POVMs. For hardness, it will suffice it consider only observations with rank 1 and trace 1, but for now we allow them to be general. For any particular $\rho$ and sequence $\mathcal{O}_i$, the likelihood $\prod \Tr[\mathcal{O}_i \rho]$ can be evaluated directly in $O(nd^2)$ operations. The difficulty then lies in the normalizing factor,

$$p_{norm} = P(\mathcal{O})$$
$$\textrm{so that}$$
$$P(\rho|\mathcal{O}) = p_{norm}^{-1}\prod_{i\in [n]} \Tr[O_i \rho]$$

This indicates the probability of an entire sequence of observations. While a single observation has the simple form of $P(\mathcal{O}_i) = \Tr[\mathcal{O}_i]$, the expression rapidly becomes more complicated as we consider sequences of observations.
\par A brief example is useful for understanding what $p_{norm}$ represents. Suppose that we measure a qubit 1000 times along each of the X, Y, and Z axes: we expect to see a particular amount of bias. Observing 1000 results each of +X, +Y, and +Z would be very unlikely, as the qubit cannot be in the +1 eigenstate of all three axes at once. It would be similarly surprising to see exactly 500 counts each of +X, -X, +Y, -Y, +Z, and -Z: this state shows no tendency of a particular orientation, but a pure qubit state must show a bias towards some orientation. This would have a small value of $p_{norm}$, as there is no good state to explain the sequence observed. A sequence of 1000 +Z observations, and 500 each of +X, -X, +Y, and -Y is much more likely, as it can be well explained by the $\ket{\uparrow}$ state, and so has a larger value of $p_{norm}$.
\par As we just saw, computing $P_{density}(\rho = \rho_0|\mathcal{O})$ is easy if $p_{norm}$ is known, and conversely $p_{norm}$ can be easily computed from the probability density. $p_{norm}$ is a more attractive goal for our problem, as it doesn't depend on $\rho_0$. It can be computed by summing up all unnormalized probabilities:

$$p_{norm} = \int_{\vec x \in \mathbb{C}^d_{1}} \prod_{i\in [n]} \Tr[\mathcal{O}_i xx^\dagger]\,dx$$

where the integral is over the Hilbert space $\mathbb{C}^d$ restricted to length-1 vectors.  This leads to the definition,
\begin{definition}[Quantum-Bayesian-Update]
Given a collection of observations $\mathcal{O} = (\mathcal{O}_1,\dots \mathcal{O}_n)$ in a Hilbert space of dimension $d$, compute
\begin{equation}\label{eqn:pnormDef}
p_{norm} = \frac{\int_{\vec x \in \mathbb{C}^d_{1}} \prod_{i\in [n]} \Tr[\mathcal{O}_i xx^\dagger]\,dx}{\Tr[\mathcal{O}_i]}
\end{equation}
\end{definition}

\subsection{Polynomial time QBU for fixed $d$}\label{sec:fixDAlgo}
This space of state vectors $\mathbb{C}_1^d$ has the geometry of a real $(2d-1)$-sphere, and the entries of $\rho$ are quadratic in the Cartesian coordinates for this sphere. Thus, $p_{norm}$ becomes a integral over a $(2d-1)$-sphere of a homogeneous $2n$ degree polynomial in the $2d$ variables. The expansion of the polynomial into monomials takes $O((2n)^{2d})$ time, and each monomial can then be immediately integrated over the sphere using the formula\cite{Folland2001}
\begin{equation}\label{sphereInt}
\int_{S^k} x_1^{\alpha_1}x_1^{\alpha_2}\dots x_k^{\alpha_k} = \begin{cases}
0 & \textrm{if any $\alpha_i$ are odd}\\
\frac{2\prod_i\Gamma(\frac{1}{2}(\alpha_i+1))}{\Gamma(\sum_i \frac{1}{2}(\alpha_i+1))} & \textrm{if all $\alpha_i$ are even}\\
\end{cases}
\end{equation}
where $\Gamma$ is gamma function, $\Gamma(\frac{1}{2}(\alpha+1)) = \sqrt{\pi 2^\alpha}(\alpha-1)!!$. This gives a polynomial time algorithm for evaluating $p_{norm}$ when $d$ is fixed.

\subsection{Relationship between estimation problems}
Since QBU is not of particular interest to actual tomography tasks, we show it is equivalent (under polynomial many-one reductions) to the more realistic tasks 2 and 3 above, of estimating observables or the state itself. We can show that these are just as difficult (or, just as easy) as the Bayesian update step.

\subsubsection{Computing $\rho_{Avg}$}
Given that there will always be room for uncertainty, we cannot meaningfully as for a single pure state as an answer, but we can ask for $\rho_{Avg}$: the mixed state representing the correctly updated mixture over all the possible true states, given by $\int P(\rho) \rho\,d\rho$. The impure $\rho_{Avg}$ reflects the expectation of all observables given our current information.
\par We parameterize the space of density matrices by a single vector $\psi \in S^{2d-1}$, and given some completed observations $\mathcal{O}$, the Bayesian expected state is
$$\rho_{Avg} = \int_{\psi \in S^{2d-1}} P\Big(\ket{\psi}\bra{\psi}\Big|\mathcal{O}\Big) \ket{\psi}\bra{\psi}\,d\psi $$
$$ = \int_{\psi \in S^{2d-1}} p_{norm}^{-1}\Big(\ket{\psi}\bra{\psi}\Big)\prod_{O\in \mathcal{O}} \braket{\psi|O|\psi} \,d\psi$$
whose individual matrix elements are
$$\braket{i|\rho_{Avg}|j} = p_{norm}^{-1}\int_{\psi \in S^{2d-1}} \braket{i|\psi}\braket{\psi|j}\prod_{O\in \mathcal{O}} \braket{\psi|O|\psi} \,d\psi$$
We have already discussed computing $p_{norm}$, as a spherical integral of a polynomial. For any given $i$ and $j$, the remaining integral is also a spherical integral of a polynomial, and can be computed in the same fashion. In fact we can re-use the results from the large product excluding the $i$ and $j$, and so $\rho_{Avg}$ can be recovered in $O(n^d)$ time.
\par On the other hand, a diagonal element $\braket{i|\rho_{avg}|i}$ gives
$$\braket{i|\psi}\braket{\psi|i}\prod_{O\in \mathcal{O}} \braket{\psi|O|\psi} = \bra{\psi}\Big(\ket{i}\bra{i}\Big)\ket{\psi}\prod_{O\in \mathcal{O}} \braket{\psi|O|\psi} = \prod_{O\in (\mathcal{O} \cup \{\ket{i}\bra{i}\})} \braket{\psi|O|\psi}$$
which is the same integrand as for $p_{norm}$, only with one additional observation $\ket{i}\bra{i}$ added.
\par If we had an algorithm compute $\rho_{Avg}$ efficiently, we could use it to solve the Bayesian update problem on a set of observations $\mathcal{O}$, by discarding the last observation $O_{last}$, computing $\rho_{Avg}$, decompose $O_{last}$ into a scaled sum of projectors $\sum_i \lambda_i \ket{i}\bra{i}$, and then evaluate the sum of matrix elements $\sum_i \lambda_i \bra{i}\rho_{Avg}\ket{i}$. This shows that state estimation is at least as hard as Bayesian updating.

\subsubsection{Computing observable expectations}
We could try to only find the expectation of a particular observable $A$, and not the whole state $\rho_{Avg}$, conditioned on our observations. We can write this as $E[A|\mathcal{O}]$. This is also just as hard: density matrices as a $d^2-1$ linear space, and expectations of observables are linear in $\rho$, so by computing the exact expectation of $d^2-1$ independent obsevables, we can find $\rho_{Avg}$ exactly. This is of course precisely the idea behind least-squares quantum state estimation, and it shows that computing expectation values is as hard as $\rho_{Avg}$.
\par Finally, if we could compute a Bayesian update, we could compute the expectation values of observables. Just as before, write our desired obsevable as $A = \sum \lambda_i \ket{i}\bra{i}$, and evaluate
$$E[A|\mathcal{O}] = \sum \lambda_i E[\ket{i}\bra{i}|\mathcal{O}] = \sum \lambda_i p_{norm}^{-1} \int_{\psi \in S^{2d-1}} \prod_{O\in (\mathcal{O} \cup \{\ket{i}\bra{i}\})} \braket{\psi|O|\psi} \,d\psi$$
Computing $p_{norm}$ and each of the $d$ many spherical integrals is a Bayesian update problem. We have reductions (Bayesian update) $\to$ (Compute $\rho_{Avg}$) $\to$ (Compute $E[A|\mathcal{O}]$) $\to$ (Bayesian update), so these are equivalent in difficulty. Note that these are many-one reductions, which is unavoidable as $\rho_{Avg}$ is a matrix-valued function problem while the two are scalar-valued.

\subsection{NP-Hardness of QBU and $\rho_{Avg}$}
We now state the main hardness results on quantum tomography.
\begin{theorem}\label{thm:qbuHard}
For any $C < 1$, it is NP-hard to compute the value $p_{norm}$ for Quantum-Bayesian-Update with an approximation factor of at most $2^{n^C}$.
\end{theorem}
\begin{proof}
When $\mathcal{O}_i$ are all rank-1 operators, the numerator in Eq. \ref{eqn:pnormDef} is of the form in Theorem \ref{thm:permIsSphere}, and the denominator in Eq. \ref{eqn:pnormDef} can be efficiently computed by direct calculation. Thus any PSD permanent can be efficiently reduced to a problem of computing $p_{norm}$ with an approximation-preserving reduction, and \textsc{QBU} is \textsf{NP}-Hard to approximate to the same degree.
\end{proof}
\begin{theorem}\label{thm:stateHard}
For any $C < 1$, it is \textsf{NP}-hard to compute a diagonal matrix entry of $\rho_{Avg}$, in any basis, with an approximation factor of at most $2^{n^C}$. It is also \textsf{NP}-hard to compute the expectation of a positive semidefinite operator $\mathcal{O}$ with an approximation factor of at most $2^{n^C}$.
\end{theorem}
\begin{proof}
A diagonal element of $\rho_{Avg}$ is the expectation value of the rank-1 PSD operator projecting onto that element, so the first statement is a special case of the second. As described above, both of these quantities then also take the form of a PSD permanent, and any PSD permanent can be turned into these problem by taking the desired matrix element (in the first case) or observabe $\mathcal{O}$ (in the second case) to be the first vector $V_1^\dagger V_1$. These are also approximation preserving reductions, so these are also \textsf{NP}-hard to approximate.
\end{proof}

\subsection{NP-Completeness of Maximum Likelihood Estimation}
In the case of MLE state tomography, we are not so demanding that we require knowledge of the full average state, and we are content with just finding one good explanatory state $\ket{\psi}$. Accordingly, we do not consider a permanent $\int_x I_V(x)$ (a problem of counting solutions to 3-SAT), but just the question of maximizing $I_V(x)$ (a problem of finding a solution to 3-SAT). This allows to show that the problem is actually lies in \textsf{NP}, while this is unlikely to be true for the other problems in this paper unless $\textsf{BPP}^\textsf{NP} = \textsf{NP}$.
\par Formulating the MLE problem as a decision problem:
\begin{definition}[$C$-Approximate-Quantum-MLE]
Given a collection of observations $\mathcal{O}_i$ of an unknown quantum state $\ket{\psi}$, and a real number $p$, decide whether there is a $\ket{\psi}$ whose likelihood $L(\psi) = \prod_i \braket{\psi|\mathcal{O}_i|\psi}$ is at least $p$, or if $L(\psi) < p/C$ for all $\psi$, being promised that one of these is the case.
\end{definition}
We will show that even the approximate problem is \textsf{NP}-hard, for any $C$.
\begin{theorem}\label{thm:MLEhard}
For any $C > 1$, the $C$-Approximate-Quantum-MLE problem is \textsf{NP}-complete.
\end{theorem}
\begin{proof}
Containment in NP is straightforward, as one can supply a description of the state $\ket{\psi}$, which requires only $d$ many real numbers, and then $L(\psi)$ can be directly evaluated.
\par To show hardness, we use the same NAE3SAT construction as in Theorem \ref{thm:FappNAE}. As was shown in the proof of that theorem, any good point (thus, a solution to the underlying NAE3SAT problem) has
$$L(\psi) = I_0(x) \ge p\left(1 - \frac{\ln^2 d}{\sqrt d}\right).$$
We also show in that proof that, if there are no good points (and thus no solutions) then
$$L(\psi) = I_0(x) \le p/d^{d^2}$$
for all points. Thus, the existence of a high likelihood point even within $C < d^{d^2}$ implies the existence of a solution.
\end{proof}

\subsection{Practical Difficulty of Tomography}
Although the above results imply that several approaches to quantum state tomography may be difficult to compute exactly, these difficult instances are somewhat artificial and unlikely to occur in practice. Additionally, difficult instances such as the one constructed in the above proofs could be readily addressed in practice by the addition of measurements in e.g. the $X$ measurement basis, which would directly probe the relative signs in the state vector and allow relatively efficient readout of the state. Additionally, the constraint that we only search for pure states -- while a useful prior that could be relevant once high-fidelity quantum computer exists -- makes a highly nonconvex search space. If we relax this and take a prior with uniform measure over the space of density matrices, then the resulting likelihood function is logarithmically convex and the resulting MLE problem can be solved in polynomial time in $d$. Thus, these results should not be taken as a statement that quantum state tomography is actually exponentially hard in the Hilbert space dimension $d$. Rather, any analysis of quantum state tomography procedures will need at least one of: careful choice of measurement basis, only probabilistic guarantees on convergence, or (if doing MLE) a convex prior.

\printbibliography

\end{document}